\documentclass[11pt]{article}

\usepackage[margin=1in]{geometry}
\usepackage{setspace}
\usepackage{amsmath,amssymb,amsthm}
\usepackage{bm}
\usepackage{graphicx}
\usepackage{color}
\usepackage{hyperref}

\usepackage{algorithm}
\usepackage{algpseudocode}

\DeclareMathOperator*{\argmax}{arg\,max}

\newtheorem{theorem}{Theorem}

\newtheorem{corollary}{Corollary}

\newcommand{\Z}{\mathbb{Z}}

\newcommand{\X}{\mathbb{X}}
\newcommand{\F}{\mathbb{F}}
\newcommand{\R}{\mathbb{R}}

\newcommand{\C}{\mathbb{C}}

\newcommand{\Var}{\mathrm{Var}}
\newcommand{\Cov}{\mathrm{Cov}}

\newcommand{\bY}{\bm{Y}}
\newcommand{\bK}{\bm{K}}
\newcommand{\bmf}{\bm{f}}
\newcommand{\bR}{\bm{R}}

\newcommand{\cY}{\mathcal{Y}}

\newcommand{\loc}{x}

\usepackage[round]{natbib}

\begin{document}

\begin{center}{\LARGE Nonparametric Spectral Methdods for  \\

\vspace{6pt}

Multivariate Spatial and Spatial-Temporal Data }

\vspace{12pt}

{Joseph Guinness}
\vspace{12pt}

\textit{Cornell University, Department of Statistical Science}
\vspace{24pt}

\textbf{Abstract}

\end{center}

We propose computationally efficient methods for estimating stationary multivariate spatial and spatial-temporal spectra from incomplete gridded data. The methods are iterative and rely on successive imputation of data and updating of model estimates. Imputations are done according to a periodic model on an expanded domain. The periodicity of the imputations is a key feature that reduces edge effects in the periodogram and is facilitated by efficient circulant embedding techniques. In addition, we describe efficient methods for decomposing the estimated cross spectral density function into a linear model of coregionalization plus a residual process. The methods are applied to two storm datasets, one of which is from Hurricane Florence, which struck the souteastern United States in September 2018. The application demonstrates how fitted models from different datasets can be compared, and how the methods are computationally feasible on datasets with more than 200{,}000 total observations.

\section{Introduction}\label{introduction}

There are several exciting scientific campaigns to produce and observe multivariate data that vary over space and time. For example, large climate centers such as the National Center for Atmospheric Research (NCAR) produce high resolution simulations of the Earth system. These models include dozens of variables evolving in concert over space and time. The National Aeronautics and Space Administration (NASA) and the National Oceanic and Atmospheric Administration (NOAA) have deployed numerous satellites that collect observations of the Earth surface and atmosphere. NASA and NOAA recently launched a pair of satellites in its ongoing geostationary operational environmental satellite (GOES) program, the GOES-16 and GOES-17 spacecraft, that sit in geostationary orbit, continually monitoring light reflected by the Earth surface and atmosphere in 16 separate wavelength bands. In turn, the raw data are processed to produce dozens of physically relevant variables, such as atmospheric water vapor content and land surface temperature.

A statistical framework for analyzing these data should, at a minimum, include (1) sufficiently flexible statistical models capable of capturing complex multivariate dependencies, and (2) computationally efficient tools for estimating the models from data.  \cite{genton2015cross} provide a thorough review of existing modeling and estimation frameworks. \cite{kleiber2017coherence} conducted a theoretical analysis of a number of multivariate spatial models from the literature and concluded that many of them impose oversimplistic restrictions on the coherence between pairs of variables. For example, separable and kernel convolution multivariate spatial models \citep{majumdar2007multivariate} have constant coherence. In the linear model of coregionalization (LMC) \citep{banerjee2014hierarchical}, when component processes have spectral densities that decay with different rates, the coherence always converges to a non-zero constant as the frequency increases.

It is desirable to have models that do not impose such restrictions, for example, to have a model class that allows the coherence to decay to zero as frequency increases.
\cite{kleiber2017coherence} showed that the multivariate Mat\'ern models \citep{gneiting2010matern,apanasovich2012valid} do possess such flexiblity. However, likelihood-based estimation of parameters in multivariate Mat\'ern models is not computationally feasible for the massive datasets mentioned in our opening paragraph. Computational issues come from two sources. The first is that the data size prohibits formation and factoring of the covariance matrix, necessary operations for evaluating the likelihood function. There exist promising approximations that could in principle apply to the multivariate case, such as Vecchia's likelihood approximation \citep{vecchia1988estimation,datta2016hierarchical,guinness2018permutation,katzfuss2017general} and hierarchical matrix approximations \citep{saibaba2012application,ambikasaran2016fast,litvinenko2017likelihood}, but these approximations are untested for multivariate models, and it is not yet clear how to optimally implement them, and how their performance will compare to univariate cases. Such work would be a welcome development to the literature. However, a second and perhaps more serious computational problem, even for fast approximate methods, is that multivariate spatial models contain many more parameters, usually on the order of the square of the number of components. The large number of parameters poses a serious issue for optimization of the (approximate) likelihood function, which typically requires many iterations until convergence over a large parameter space. We demonstrate this problem with a simple example in Section \ref{simulationsection}

Nonparametric spectral methods offer the potential of addressing both the modeling and computational demands for multivariate spatial and spatial-temporal data. In this framework, a discrete Fourier transform (DFT) is applied to each multivariate component, and then the periodogram vectors are smoothed over frequency to generate a nonparametric estimate of the cross spectral density matrices. This approach is reasonably flexible from a modeling standpoint, since the only assumption is that the cross spectral density matrices vary smoothly with frequency, and it is computationally inexpensive, since we can evaluate the DFTs quickly with FFT algorithms. However, nonparametric spectral methods have limited applicability since they typically apply only to complete, gridded data on rectangular domains. Moreover, even if we have such data, edge effects can introduce severe biases in the spatial and especially the spatial-temporal cases \citep{guyon1982parameter}. Tapering \citep{dahlhaus1987edge} and differencing \citep{lim2008properties} have the potential to reduce edge effects, but differencing is not possible when there are missing values, and \cite{guinness2017spectral} showed that tapering can be ineffective under certain missingness scenarios and suboptimal compared to imputation-based methods.

This paper provides methods that addresses edge effects and extend the applicability of nonparametric spectral methods to incomplete gridded multivariate spatial-temporal data. This is achieved by leveraging the simple, yet powerful technique of periodic imputation introduced for the univariate case in \cite{guinness2017spectral}. We also introduce a method for decomposing the estimated cross spectral density function into an LMC cross spectral density function plus a residual cross spectral density function, which can be used as tool for exploring the spatial-temporal variation in the data. We apply the methods to compare the multivariate spatial-temporal covariances of data from an ordinary storm and from Hurricane Florence, demonstrating that these methods can computationally feasibly estimate multivariate spatial-temporal models with more than 200{,}000 obervations, and how the estimates can be interpreted to distinguish among models.

\section{Iterative Spectrum Estimator}

\subsection{Model and Notation}

Let $\loc \in \Z^d$ be a location on the $d$-dimensional integer lattice, and let $\bY(\loc) = (Y_1(\loc),\ldots,Y_p(\loc)^T) \in \R^p$ be a $p$-variate mean-zero random process on $\Z^d$. The process $\bY$ is stationary if $\loc_1 - \loc_2 = \loc_3 - \loc_4$ implies that $E( \bY(\loc_1)\bY(\loc_2)^T ) = E( \bY(\loc_3)\bY(\loc_4)^T )$, in which case we can define $\bK( h ) = E( \bY(\loc)\bY(\loc+h)^T )$ as the covariance between the process at pairs of locations separated by lag $h = (h_1,\ldots,h_d)$.
Cramer's Theorem \citep[cf.]{wackernagel1998multivariate} states that $\bK$ has the representation \begin{align}
\bK(h) = \int_{[0,1]^d} \bmf(\omega) \exp(2\pi i\omega \cdot h)d\omega,
\end{align}
where $i = \sqrt{-1}$, and $\omega \cdot h = \omega_1 h_1 + \cdots + \omega_d h_d$. The matrix-valued function $\bmf(\omega) \in \C^{p \times p}$ has $(j,k)$ entry $f_{jk}(\omega)$, and must be a Hermitian positive definite matrix in order that $\bK(h)$ is a real and positive definite function. We call $\bmf(\omega)$ the cross-spectral density (CSD) function.

Let $j_1,\ldots,j_n \in \{1,\ldots,p\}$ and $\loc_1,\ldots,\loc_n$ be a sequence of locations and $U = (Y_{j_1}(\loc_1),\ldots,Y_{j_n}(\loc_n))$ be the vector of observations. In this notation, $n$ is the total number of individual observations, and we may observe between $0$ and $p$ of the components at any location. Because the process is stationary, we can assume without loss of generality that each location $x_k$ falls ``northeast'' of the location $(1,\ldots,1)$ and ``southwest'' of some location $a = (a_1,\ldots,a_d)$. Formally we say that $x_k$ falls in the observation lattice  $\X_a$, where
\begin{align}
\X_a = \{ (r_1,\ldots,r_d) \, : \, r_k \in \{1,\ldots,a_k\} \mbox{ for every } k  \}.
\end{align}
The methods in this paper involve imputing the multivariate process onto a domain larger than the observation lattice $\X_a$. To this end, define $b = (b_1,\ldots,b_d)$ with $b_j > a_j$ for each $j$, and the corresponding embedding lattice $\X_b$. Let $(U,V)$ denote a complete vector of observations on $\X_b$, that is $(U,V)$ contains each of the $p$ elements of $\bY(\loc)$ at each location $\loc \in \X_b$, ordered with the observations $U$ first, then the missing values $V$ second. Figure \ref{latticefigure} provides a visual description of these definitions.

\begin{figure}
\centering
\includegraphics[width=\textwidth]{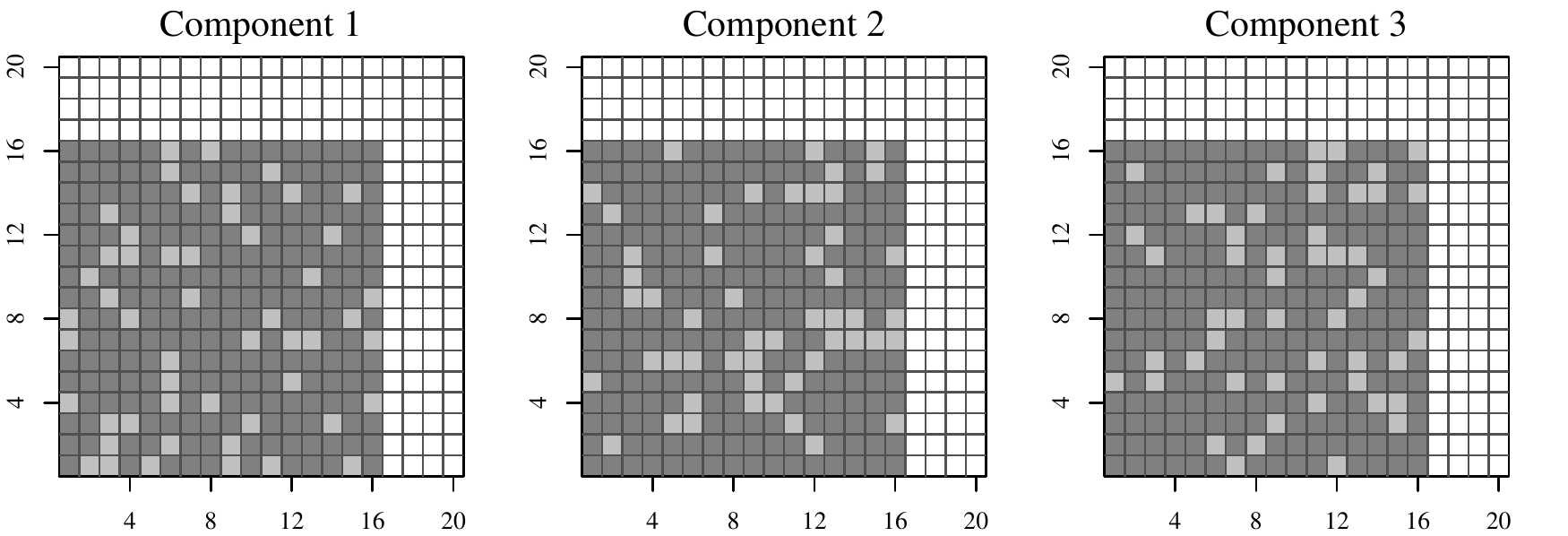}
\caption{ \label{latticefigure} For a $p=3$ component dataset, observation lattice $\X_a$ (light + dark gray) and embedding lattice $\X_b$ (gray + white), for $a = (16,16)$ and $b = (20,20)$. Dark gray represents locations for observations in $U$, light gray and white represent locations for observations in $V$. Note that locations of missing data on $\X_a$ need not be the same in each component. }
\end{figure}

In the next subsection, we describe methods for iteratively imputing the missing values $V$ on $\X_b$. However, we first describe how we obtain a CSD estimate given a complete set of observations $(U,V)$ on $\X_b$. We propose the following routine for obtaining a CSD estimate $\widehat{\bm{f}}[(U,V)]$:
\begin{align}
\label{step1}  \cY_j(\omega) &= \frac{1}{\sqrt{m}} \sum_{\loc \in \X_b} Y_j(\loc) \exp(-2\pi i \omega \cdot \loc) \\
\label{step2} \theta_j &= \argmax_\theta \sum_{\omega \in \F_b} \left[  - \log f_\theta(\omega) - \frac{ |\cY_j(\omega)|^2 }{f_\theta(\omega)} \right] \\
\label{step3} \widehat{f}_{jk}(\omega) &= \big( {f}_{\theta_j}(\omega) {f}_{\theta_k}(\omega) \big)^{1/2}
\sum_{\nu \in \F_b} \frac{ \cY_j(\nu) \cY_k(\nu)^* }{\big( {f}_{\theta_j}(\nu) {f}_{\theta_k}(\nu) \big)^{1/2} } \alpha(\omega - \nu).
\end{align}
The symbol $\F_b$ refers to the set of Fourier frequencies on a grid of size $b$, and $*$ is complex conjugate. Each computation written in terms of $\omega$ is computed for every $\omega \in \F_b$. If the equation is written in terms of $j$, it is computed for each $j = 1,\ldots,p$, and likewise, equations written in terms of $j$ and $k$ are computed also for each $k = 1,\ldots,p$. In \eqref{step1}, we compute the DFT of the complete dataset on $\X_b$. In \eqref{step2}, we maximize Whittle's loglikelihood approximation for some parametric spectral density $f_\theta$. In Sections \ref{simulationsection} and \ref{datasection} we use the quasi Mat\'ern spectral density proposed in \cite{guinness2017circulant}. In \eqref{step3}, we smooth the sample covariances of the normalized DFT entries with kernel $\alpha$, and multiply by the square roots of the estimated parametric spectral densities.

\cite{guinness2017spectral} suggested the parametric normalization in \eqref{step3} as a method for reducing smoothing bias, which allows for the use of wider smoother kernels, resulting in estimates with smaller variance. If instead $f_\theta$ is assumed to be constant, the routine for obtaining a CSD estimate reduces to a standard routine described, for example, in \cite{brockwelldavis} for multivariate time series. Nonparametric or semiparametric CSD estimation methods are desirable because the estimated CSD function $\widehat{\bm{f}}(\omega)$ is automatically positive definite at every frequency if the support of $\alpha$ includes at least $p$ Fourier frequencies in $\F_b$. However, when the dimension $d$ of the field is larger than 1, substantial bias is imparted by edge effects \citep{guyon1982parameter,lim2008properties}. In the next subsection, we describe iterative imputation methods for reducing edge effect bias.

\subsection{Iterative Periodic Imputation Methods}

\cite{lee2009nonparametric} proposed a method for estimating a univariate time series spectrum from incomplete data by iteratively imputing the missing data on the observation domain.
\cite{guinness2017spectral} showed that when $d > 1$, univariate spectral density estimates can be improved dramatically when data are iteratively imputed onto embedding lattice $\X_b$ with a conditional simulation from a periodic model on $\X_b$. Here, we propose an extension of the periodic imputation method to the estimation of multivariate cross spectral densities. At iteration $\ell$, the missing values $V$ are imputed according to a Gaussian process model with a covariance function that is periodic on $\X_b$,
\begin{align}
\bR^{(\ell)}(h) = \frac{1}{m} \sum_{\omega \in \F_b } \bmf^{(\ell)}(\omega) \exp(2\pi i \omega \cdot h),
\end{align}
where $\bm{f}^{(\ell)}$ is the CSD function at the $\ell$th iteration, and $m = b_1 \cdots b_d$ is the total number of locations in $\X_b$. With the complete dataset $(U,V)$, we update the spectrum estimate with $\bm{f}^{(\ell+1)} = \widehat{\bm{f}}[(U,V)]$ and iterate. After a burn-in period of $B$ iterations, we update with a weighted average of the previous and current spectrum, and we monitor changes in $\bm{f}^{(\ell)}$ for convergence. The full estimation algorithm is as follows:

\vspace{11pt}

\noindent \textbf{Estimation Algorithm}:

\vspace{6pt}

\noindent Initialize with $\bmf^{(1)}(\omega) = \widehat{\bm{f}}[(U,0)]$. For $\ell = 1,\ldots,B$,
\begin{enumerate}
\itemsep0pt
\item Simulate $V$ given $U$ from Gaussian process with covariance $\bR^{(\ell)}$.
\item Given $(U,V)$, set $\bm{f}^{(\ell+1)} = \widehat{\bm{f}}[(U,V)]$.
\end{enumerate}

%\noindent \textbf{Convergence monitoring period}:

\noindent For $\ell > B$
\begin{enumerate}
  \itemsep0pt
\item Simulate $V$ given $U$ from Gaussian process with covariance $\bR^{(\ell)}$.
\item Given $(U,V)$, set
\begin{align*}
\bm{f}^{(\ell+1)} = \frac{\ell-B-1}{\ell-B} {\bm{f}}^{(\ell)} + \frac{1}{\ell-B}\widehat{\bm{f}}[(U,V)].
\end{align*}
\item Stop if
\begin{align*}
\max_{j, \omega} \frac{ |f_{jj}^{(\ell+1)}(\omega) - f_{jj}^{(\ell)}(\omega)|}{ f^{(\ell)}_{jj}(\omega) } < \varepsilon
\end{align*}
and return CSD estimate $\widehat{\bm{f}} = \bm{f}^{(\ell+1)}$.
\end{enumerate}

\vskip11pt

Drawing Gaussian $V$ given $U$ is the main computational challenge for the iterative methods. All of the other operations are either pointwise multiplications, divisions, DFTs, or convolutions, which can be computed with a DFT. Drawing $V$ given $U$ involves an unconditional simulation from $\bm{R}^{(\ell)}$ on $\X_b$, which can be done with DFTs and pointwise multiplications. It also involves solving a linear system with the covariance matrix for $U$ under $\bm{R}^{(\ell)}$. This is the most demanding step. We use preconditioned conjugate gradient (PCG), which is an iterative method for solving positive definite linear systems \citep{saad2003iterative}. The forward multiplications in PCG can be done in $O(p m \log m + p^2 m)$ with circulant embedding techniques \citep{wood1994simulation}. We use a preconditioner based on Vecchia's approximation \citep{vecchia1988estimation} that can be computed in $O(pm)$ time.

\section{Factor Decompositions}

The CSD $\bm{f}(\omega)$ has an interpretation as the covariance matrix for $\bm{\cY}(\omega)$, the variation in $\bm{Y}$ at frequency $\omega$. This is admittedly difficult to communicate to a diverse audience. One solution is to simply report the estimated cross covariance function, the inverse DFT of $\widehat{\bm{f}}$. To further facilitate interpretability of a $p$-dimensional multivariate model, it can be useful to consider lower-dimensional representations for the CSD functions. To this end, we propose a factor model consisting of the linear model of coregionalization (LMC) plus a residual process $\bm{Z}(x)$,
\begin{align}
\bm{Y}(x) = \sum_{j=1}^J \bm{A}_j W_j(x) + \bm{Z}(x),
\end{align}
where $\bm{A}_j \in \R^p$, $W_1,\ldots,W_J$ are univariate process components of the LMC, independent of each other and of $\bm{Z}$, a $p$-variate spatial process with correlated components. The CSD function for this process is
\begin{align}\label{factorCSD}
\bmf(\omega) = \sum_{j=1}^J \bm{A}_j \bm{A}_j^T g_j(\omega) + \bm{h}(\omega),
\end{align}
where $g_j(\omega) > 0$ is the univariate spectral density for $W_j$, and $\bm{h}$ is the CSD function for $\bm{Z}$.

There is a clear lack of identifiability in assigning power from $\bm{f}(\omega)$ to either $\sum \bm{A}_j \bm{A}_j^T g_j(\omega)$ or $\bm{h}(\omega)$; for example, any model written
as \eqref{factorCSD} can be rewritten as $\bm{f}(\omega) = \bm{h}_1(\omega)$, where $\bm{h}_1(\omega) = \sum \bm{A}_j \bm{A}_j^T g_j(\omega) + \bm{h}(\omega)$. For this reason, we argue that it is not meaningful to ask for an ``optimal'' or ``true'' value of $J$. We take the position here that since the purpose of the factor model is to explain variation in the response with a lower-dimensional model, we should assign as much power as possible to the factor term over a range of values of $J$. Then the amount of variation explained by the factor term can be used as one way to distinguish among models.

\subsection{Profiled Optimization}

After the CSD function $\widehat{\bm{f}}$ has been estimated using the proposed periodic imputation methods, we aim to select $\bm{A}_1,\ldots,\bm{A}_J$ and $g_1,\ldots,g_J$ in order to minimize the amount of power assigned to the residual process $\bm{Z}$. To quantify the power, we consider the sum of the variances of the residual processes, which can be written in terms of the sum of traces of the residual CSD function,
\begin{align}
\sum_{k=1}^p \Var( Z_k(x) ) = \frac{1}{m} \sum_{\omega \in \F_b} \mbox{Tr}\big(\bm{h}(\omega)\big) = \frac{1}{m}\sum_{\omega \in \F_b} \mbox{Tr} \Big( \widehat{\bm{f}}(\omega) - \sum_{j=1}^J \bm{A}_j \bm{A}_j^T g_j(\omega) \Big),
\end{align}
%and we define $( \widehat{\bm{A}}_1,\ldots,\widehat{\bm{A}}_J, \widehat{g}_1,\ldots,\widehat{g}_j )$ as the minimizers of the trace loss,
subject to $g_j(\omega) > 0$ and $\widehat{\bm{f}}(\omega) - \sum \bm{A}_j\bm{A}_j^T g_j(\omega)$ nonnegative definite for each $\omega$. Even for $J = 1$, minimizing the loss function over all $g_j(\omega)$ is demanding because there are potentially thousands or millions individual frequencies $\omega \in \F_b$. This intractability has motivated a theoretical study to find closed-form minimizers over $g_j(\omega)$ for fixed $\bm{A} = (\bm{A}_1,\ldots,\bm{A}_J)$. This allows us to solve the global optimization problem by profiling out $g_j(\omega)$ and numerically minimizing the sum-of-variance criterion over $\bm{A}$. Further, we can assume without loss of generality that each $\bm{A}_j$ lies on the unit sphere because the radius of $\bm{A}_j$ can be absorbed into $g_j(\omega)$, so the numerical optimization is over a $J(p-1)$ space.

The sum-of-variance criterion can be rewritten as
\begin{align}
\frac{1}{m} \sum_{\omega \in \F_b} \sum_{k=1}^p \widehat{f}_{kk}(\omega) - \frac{1}{m} \sum_{\omega \in \F_b} \sum_{k=1}^p \sum_{j=1}^J A_{jk}^2 g_j(\omega) = \frac{1}{m} \sum_{\omega \in \F_b} \sum_{k=1}^p \widehat{f}_{kk}(\omega) - \sum_{\omega \in \F_b} \sum_{j=1}^J g_j(\omega) \sum_{k=1}^p A_{jk}^2.
\end{align}
The first term on the right does not depend on $g_j$ or $\bm{A}_j$, and so minimizing the criterion corresponds to maximizing the second term subject to $g_j(\omega) > 0$ and $\bm{h}(\omega)$ nonnegative definite for each $\omega$. For fixed $\bm{A}$ with $\sum_k A_{jk}^2 = \| \bm{A}_j \|_2 = 1$, this corresponds to maximizing $\sum_\omega \sum_j g_j(\omega)$. Since the nonnegative definiteness criterion applies separately to each $\omega$, we can maximize $\sum_j g_j(\omega)$ separately for each $\omega$.

The theorem gives a result for finding the minimizers over $g_1(\omega), \ldots, g_J(\omega)$ for fixed $\bm{A}$. Following the theorem are two corollaries that give specific results for the cases $J=1$ and $J=2$. To state the theorem, parameterize nonnegative $g_j(\omega)$ as $\exp( c_j )$, and define the $J \times J$ matrices
\begin{align}
C = \mbox{diag}( e^{c_1},\ldots,e^{c_J} ), \quad B = \bm{A}^T \widehat{\bm{f}}(\omega)^{-1} \bm{A}.
\end{align}
Further, define the Lagrangian function
\begin{align}
L( c_1,\ldots,c_J,\lambda) = e^{c_1} + \cdots e^{c_J} + \lambda \det( C^{-1} - B ).
\end{align}
\begin{theorem}\label{general_minimizer}
For fixed $\bm{A} = (\bm{A}_1,\ldots,\bm{A}_J)$, the minimizer of
\begin{align*}
\mbox{Tr}( \bm{h}(\omega) ) = \mbox{Tr}\Big( \widehat{\bm{f}}(\omega) - \sum_{j=1}^J e^{c_j} \bm{A}_j \bm{A}_j^T \Big)
\end{align*}
with respect to $c_1,\ldots,c_J$, subject to $\bm{h}(\omega)$ nonnegative definite, is either a solution to $\nabla L = \bm{0}$ or for some $k \in \{1,\ldots,J\}$, the minimizer of
\begin{align*}
\mbox{Tr}\Big( \widehat{\bm{f}}(\omega) - \sum_{j\neq k} e^{c_j} \bm{A}_j \bm{A}_j^T \Big).
\end{align*}
\end{theorem}
\begin{proof}

% proof of Theorem 1
Assume without loss of generality that $\| \bm{A}_j \|_2 = 1$ for every $j$. We first show that $\bm{h}(\omega)$ is nonnegative definite but not stricly positive definite at the minimizer. To establish a contradiction, suppose that $\bm{h}(\omega)$ is strictly positive definite at the minimizer. Thus there exists $\varepsilon > 0$ such that for any $\| \bm{u} \|_2 = 1$,
\begin{align*}
\bm{u}^T ( \widehat{\bm{f}}(\omega) - \sum_{j=1}^J \bm{A}_j \bm{A}_j^T e^{c_j} ) \bm{u} > \varepsilon.
\end{align*}
Now consider the quadratic form
\begin{align*}
\bm{u}^T ( \widehat{\bm{f}}(\omega) - \sum_{j=1}^J \bm{A}_j \bm{A}_j^T ( e^{c_j} + \varepsilon/J ) ) \bm{u} > \varepsilon - \frac{\varepsilon}{J} \sum_{j=1}^J (\bm{u}^T \bm{A}_j )^2
\geq \varepsilon - \varepsilon = 0.
\end{align*}
This establishes that each $e^{c_j}$ could have been increased by $\varepsilon/J$, and thus were not the minimizers of the trace, giving a contradiction.

This means that the determinant of $\bm{h}(\omega)$ is zero at the minimizer. Consider the matrix
\begin{align*}
\begin{bmatrix}
\widehat{\bm{f}}(\omega) & \bm{A} \\
\bm{A}^T & C^{-1}
\end{bmatrix}.
\end{align*}
Its determinant must be zero because it is equal to $\det(C^{-1})\det(\bm{h}(\omega))$. Its determinant is also equal to $\det( \widehat{\bm{f}}(\omega) ) \det( C^{-1} - B)$. We know that $\det(\widehat{\bm{f}}(\omega)) > 0$ because $\widehat{\bm{f}}(\omega)$ is strictly positive definite; therefore $\det( C^{-1} - B) = 0$ at the minimizer.

The theorem follows from using the method of Lagrange multipliers, that is, we seek to maximize $\sum_{j=1}^J e^{c_j}$ subject to $\det( C^{-1} - B ) = 0$, and checking the endpoints $e^{c_j} = 0$.

\end{proof}

The theorem states that the minimizers either solve the gradient of the Lagrangian or are the solution to a subproblem. This arises because we must check the boundaries of the parameter space when using the method of Lagrange multipliers. We must reparameterize $g_j(\omega) = \exp(c_j)$ because it is possible that the maximizer of $\sum_{j} g_j(\omega)$ subject to $\bm{h}(\omega)$ nonnegative definite includes one or more negative $g_j(\omega)$. The following two corollaries give concrete solutions for the cases $J=1$ and $J=2$.
\begin{corollary}\label{golemma1}
For fixed $\bm{A}_1$, $\mbox{Tr}( \bm{f}(\omega) - \bm{A}_1\bm{A}_1^T g_1(\omega) )$ is minimized by $g_1(\omega) = ( \bm{A}^T \bm{f}(\omega)^{-1}\bm{A}_1^T)^{-1}$.
\end{corollary}
\begin{proof}
When $J=1$, $C^{-1} = e^{-c_1}$, and $B = B_{11} = \bm{A}_1^T \widehat{\bm{f}}(\omega)^{-1} \bm{A}_1$. Therefore, $\det(C^{-1}-B) = 0$ iff $e^{-c_1} = B_{11}$ iff $e^{c_1} = 1/B_{11}$.
\end{proof}
\begin{corollary}\label{golemma2}
For fixed $\bm{A}_1,\bm{A}_2$, $\mbox{Tr}( \bm{f}(\omega) - \bm{A}_1 \bm{A}_1^T g_1(\omega) - \bm{A}_2 \bm{A}_2^T g_2(\omega) )$ is minimized by either
\begin{enumerate}
\item $g_1(\omega) = 0$ and $g_2(\omega) = 1/B_{22}$,
\item $g_2(\omega) = 0$ and $g_1(\omega) = 1/B_{11}$,
\item Or
\begin{align*}
g_1(\omega) &= \frac{B_{22} - \sqrt{B_{12}B_{21}}}{\det(B)}\\
g_2(\omega) &= \frac{B_{11} - \sqrt{B_{12}B_{21}}}{\det(B)}.
\end{align*}
\end{enumerate}
\end{corollary}
\begin{proof}

% proof of corollary 2
Setting $\nabla L = \bm{0}$ yields the three equations
\begin{align*}
e^{c_1} + \lambda [ e^{-c_1}(e^{-c_2} - B_{22}) ] &= 0 \\
e^{c_2} + \lambda [ e^{-c_2}(e^{-c_1} - B_{11}) ] &= 0 \\
(e^{-c_1} - B_{11})(e^{-c_2} - B_{22}) - B_{12} B_{21} &= 0.
\end{align*}
Eliminating $\lambda$ gives the following two equations
\begin{align*}
e^{c_1}( 1 - B_{11}e^{c_1}) = e^{c_2}( 1 - B_{22}e^{c_2}) \\
(e^{-c_1} - B_{11})(e^{-c_2} - B_{22}) - B_{12} B_{21} &= 0.
\end{align*}
One can verify that the solution is
\begin{align*}
e^{c_1} &= \frac{B_{22} - \sqrt{B_{12}B_{21}}}{\det(B)}\\
e^{c_2} &= \frac{B_{11} - \sqrt{B_{12}B_{21}}}{\det(B)}.
\end{align*}
According to Theorem 1, the solution is either the above, or the solution to one of the one-factor problems, establishing the corollary.

\end{proof}

\subsection{Conditional Expectation of Component Processes}

The conditional expectations of the $j$th factor component given the data is
\begin{align*}
E( W_j(x) | U ) &= E_{V|U}[ E( W_j(x) | U, V ) ] \\
 &= E_{V|U} \bigg[ E\bigg( \frac{1}{\sqrt{m}} \sum_{\omega \in \F_b} \mathcal{W}_j(\omega) e^{i\omega \cdot x} \, \bigg| \bm{\mathcal{Y}}(\F_b) \bigg) \bigg] \\
 &=  \frac{1}{\sqrt{m}} E_{V|U} \bigg[  \sum_{\omega \in \F_b} E\Big( \mathcal{W}_j(\omega) | \bm{\mathcal{Y}}(\omega) \Big)  \bigg] e^{i\omega \cdot x } \\
&= \frac{1}{\sqrt{m}}   \sum_{\omega \in \F_b} g_j(\omega) \bm{A}_j^T \bm{f}(\omega)^{-1}  E \big[ \bm{\mathcal{Y}}(\omega) \, | \, U \big] e^{i\omega \cdot x }.
\end{align*}
The conditional expectation of $\bm{Y}(\omega)$ given $U$ is
\begin{align*}
E(\bm{\cY}(\omega) | U ) &= E\Big( \frac{1}{\sqrt{m}} \sum_{x \in \X_b} \bm{Y}(x) e^{-i\omega \cdot x} \Big| U \Big) = \frac{1}{\sqrt{m}} \sum_{x \in \X_b} E( \bm{Y}(x) | U ) e^{-i\omega \cdot x},
\end{align*}
where $\mathcal{W}_j$ is the DFT of $W_j$. Therefore, in order to compute the conditional expectation of $W_j$, we simply compute $E(V|U)$, then take the DFT of $(U, E(V|U))$ to obtain $E( \bm{\cY}(\omega)|U)$, and then take the inverse DFT of
\begin{align*}
g_j(\omega)\bm{A}_j^T \bm{f}(\omega)^{-1}E(\bm{\cY}(\omega)|U).
\end{align*}
We use the expected factors in Section \ref{datasection} to explore the decomposition of the estimated spectrum from a thunderstorm dataset.

\section{Simulations}\label{simulationsection}

The purpose of the simulation study is to demonstrate some of the computational issues with maximum likelihood estimation of the multivariate Mat\'ern parameters, and show how our proposed methods are capable of producing fast and accurate estimates of the CSD function, even when we do not assume knowledge of the parametric form of the model.

The cross covariances in the multivariate Mat\'ern model introduced by \cite{gneiting2010matern} can be written in terms of the Mat\'ern function,
\begin{align}
\Cov(Y_j(x), Y_k(x+h)) = K_{jk}(h) = \sigma_{jk} \mathcal{M}( \| h \| \alpha_{jk} \, ; \, \nu_{jk} ), \quad \mathcal{M}(r; \nu) = \frac{ r^{\nu} }{2^{\nu-1} \Gamma(\nu) } \mathcal{K}_{\nu}(r);
\end{align}
$\mathcal{M}(r \, ; \nu)$ is the Mat\'ern function, written in terms of $\mathcal{K}_\nu$, a modified Bessel function of the second kind, with parameter $\nu$.
%The Mat\'ern function is considered to be a flexible model for marginal covariances, so it is reasonable to consider it as a model for the cross-covariances as well.
The parameters $\sigma^2_{jj}$, $\alpha_{jj}$, and $\nu_{jj}$ represent the marginal variance, the inverse range, and the smoothness of $Y_j$. These marginal parameters must all be positive. To ensure positive definiteness of the multivariate Mat\'ern, \cite{gneiting2010matern} give conditions on the parameters in the bivariate Mat\'ern model, and \cite{apanasovich2012valid} give conditions on the parameters in a more general $p$-variate setting.

In the simulation study, we consider $p=2, 3$, and $4$. The inverse range and smoothness parameters are $\alpha_{jk} = 0.25$ and $\nu_{jk} = 0.5 + 0.5(j+k-2)/(2p-2)$. To define the variance parameters, let $\beta$ be a $p\times p$ matrix with $\beta_{jk} = 0.8^{|j-k|}$. The variance parameters are
\begin{align}
\sigma_{jk} = jk\frac{\Gamma(\nu_{jk})}{\Gamma(\nu_{jk}+1)} \frac{\sqrt{\Gamma(\nu_{jj}+1)\Gamma(\nu_{kk}+1)}}{\sqrt{\Gamma(\nu_{jj})\Gamma(\nu_{kk})}}\beta_{jk}.
\end{align}
This particular model parameterization is from the parsimonious multivariate Mat\'ern family \citep{gneiting2010matern}.  Table \ref{parametertable} presents the specific parameter values used when $p=3$. Due to the symmetry constraints $\sigma_{jk} = \sigma_{kj}$, $\alpha_{jk} = \alpha_{kj}$, and $\nu_{jk} = \nu_{kj}$, the full multivariate Mat\'ern has $3p(p+1)/2$ unique parameters. For $p=2$, this is 6 parameters; for $p = 3$, 18 parameters; for $p=4$, 30 parameters. A mentioned in Section \ref{introduction}, the large number of free parameters is one aspect of the computational difficulty of working with the multivariate Mat\'ern model.

We specify grid size $(16,16)$ and simulate 70 datasets for each value of $p \in \{2,3,4\}$ with no missing values on the grid. To estimate multivariate Mat\'ern parameters, we use the exact Gaussian likelihood function, maximized using the \verb!optim! function in R with the default Nelder-Mead algorithm. The estimation procedure makes no assumptions that the true model is a member of the parsimonious multivariate Mat\'ern family. We run the Nelder-Mead algorithm for $6000$ iterations, stopping at each $1000$ iterations to report estimation progress and elapsed time.

\begin{table}
\centering

\begin{tabular}{ccc|ccc|ccc}
\multicolumn{3}{c}{Variance $\sigma_{jk}$} & \multicolumn{3}{c}{Inverse Range $\alpha_{jk}$} & \multicolumn{3}{c}{Smoothness $\nu_{jk}$} \\
\hline
1.00 & 1.57 & 1.81 & 0.25 & 0.25 & 0.25 & 0.500 & 0.625 & 0.750 \\
1.57 & 4.00 & 4.75 & 0.25 & 0.25 & 0.25 & 0.625 & 0.750 & 0.875 \\
1.81 & 4.75 & 9.00 & 0.25 & 0.25 & 0.25 & 0.750 & 0.875 & 1.000
\end{tabular}
\caption{ \label{parametertable} Parameter matrices for simulation study when $p=3$. }
\end{table}

\begin{figure}
\centering
\includegraphics[width=0.8\textwidth]{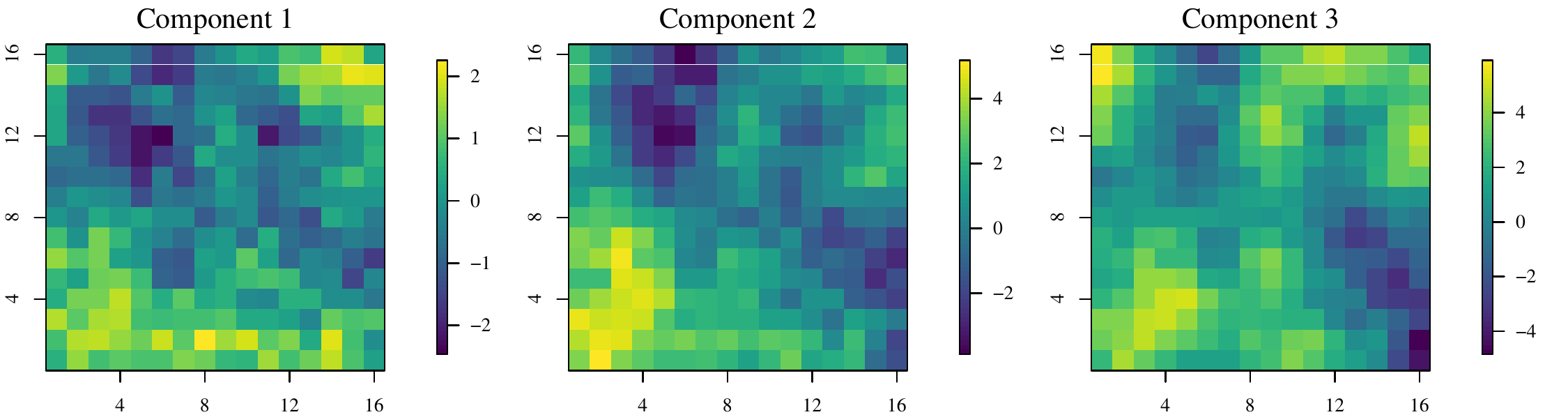}
\caption{ \label{multi_matern_sim} Example realization from multivariate Mat\'ern. }
\end{figure}

For our proposed periodic imputation methods, we take three choices of expansion parameter $\tau \in \{1.00, 1.25, 1.50\}$, and we use the parametric variant of the iterative algorithm, run for 50 burn-in iterations, and an averaging tolerance of $\varepsilon = 0.01$. We use a Gaussian smoothing kernel, with four bandwidth choices in $\{0.15,0.20,0.25,0.30\}$; a bandwidth of $b$ corresponds to $(100b)\%$ of the frequency domain. For the parametric filter, we use the quasi-Mat\'ern spectral density
\begin{align}
f_\theta(\omega) = \sigma^2\Big( 1 + \alpha^{-2} \big( \sin^2(\omega_1/2) + \sin^2(\omega_2/2) \big) \Big)^{-\nu - d/2}.
\end{align}
The parameter $\sigma^2$ can be profiled out, so at each iteration, Whittle's likelihood is optimized over the two-dimensional parameter space $(\alpha,\nu)$ for each of the $p$ components.

To evaluate the estimators, we consider a spectral norm criterion
\begin{align}
 \frac{1}{m} \sum_{\omega \in \F_b} \lambda_{\max}\Big[ \bm{f}^{-1/2}(\omega) \Big( \widehat{\bm{f}}(\omega) - \bm{f}(\omega) \Big) \bm{f}^{-1/2}(\omega) \Big]
\end{align}
where $\lambda_{\max}(A)$ is the largest absolute eigenvalue of matrix $A$, $\bm{f}$ is the true CSD function, and $\widehat{\bm{f}}$ is the estimate.

The results of the simulation study are given in Table \ref{simstudy_results}. When no periodic embedding is used ($\tau = 1.0$), the estimates are computed extremely fast but are poor relative to the other estimators. The poor performance when $\tau=1.0$ highlights the fact that spectral estimators based on the periodogram of the observed data can suffer from severe edge effects. The periodically embedded estimators are nearly as accurate or more accurate than maximum likelihood, even after 6000 Nelder-Mead iterations. When $p=2$ the periodic imputation estimators are slightly worse in terms of spectral norm, whereas the periodic imputation estimators are better when $p=3$ and $p=4$, even though no knowledge of the true CSD functions is assumed. Further, the periodic imputation methods are much faster. They converge on average in less than 25 seconds in every case, whereas the maximum likelihood estimates converge in several thousand iterations and take on the order of one to three hours.

\begin{table}
\centering
\begin{tabular}{ccccc|cccccc}
 && \multicolumn{3}{c}{Periodic Imputation} & \multicolumn{6}{c}{Max.\ Lik.\ Multivariate Mat\'ern } \\
 && \multicolumn{3}{c}{Expansion Factor $\tau$} & \multicolumn{6}{c}{Nelder-Mead Iterations} \\
components & quantile & 1.00 & 1.25 & 1.50 & 1000 & 2000 & 3000 & 4000 & 5000 & 6000 \\
 \hline
      &  0.75 &   2.540 &    0.304 &    0.295 &   0.341 &    0.275 &    0.272 &    0.245 &    0.242 &    0.239  \\
$p=2$ &  0.50 &   1.630 &    0.276 &    0.266 &   0.277 &    0.218 &    0.208 &    0.204 &    0.204 &    0.201  \\
      &  0.25 &   1.206 &    0.253 &    0.243 &   0.200 &    0.181 &    0.170 &    0.173 &    0.170 &    0.172  \\
      &  time & 0.078 &  0.227 &  0.261 &  16.2 &   33.1 &   50.0 &   66.5 &   83.0 &   98.3  \\
   \hline
      &  0.75 &   2.494 &    0.379 &    0.358 &   0.958 &    0.563 &    0.526 &    0.504 &    0.487 &    0.488  \\
$p=3$ &  0.50 &   1.767 &    0.337 &    0.328 &   0.686 &    0.487 &    0.455 &    0.441 &    0.440 &    0.428  \\
      &  0.25 &   1.300 &    0.306 &    0.302 &   0.579 &    0.415 &    0.392 &    0.386 &    0.361 &    0.360  \\
      &  time & 0.081 &  0.273 &  0.296 &  22.2 &   44.4 &   67.4 &   89.8 &  112.2 &  138.3  \\
   \hline
      &  0.75 &   2.654 &    0.429 &    0.407 &   2.847 &    1.197 &    0.834 &    0.801 &    0.719 &    0.681  \\
$p=4$ &  0.50 &   1.888 &    0.398 &    0.383 &   2.096 &    0.927 &    0.720 &    0.667 &    0.605 &    0.594  \\
      &  0.25 &   1.412 &    0.372 &    0.358 &   1.627 &    0.713 &    0.587 &    0.535 &    0.521 &    0.504  \\
      &  time & 0.090 &  0.363 &  0.416 &  29.8 &   59.8 &   90.2 &  118.1 &  147.4 &  175.9 \\
   \hline
\end{tabular}
\caption{\label{simstudy_results} Simulation study results, showing 0.25, 0.50 and 0.75 quantiles of spectral norms over 70 simulation replicates. Periodic imputation methods evaluated at three expansion factors and with kernel bandwidth 0.30. Maximum likelihood conducted in R with optim function using Nelder-Mead algorithm. Time is average time in minutes over 70 simulation replicates.}
\end{table}

The simulations were conducted on a single node of the Cheyenne supercomputer managed by the National Center for Atmospheric Research. Each node has 36 2.3-GHz Intel Xeon E5-2697V4 processors, and each processor has 2 threads. One processor was used to manage the parallel tasks, and each of the remaining 35 processors analyzed two datasets in parallel on the two threads. The computing times in Table \ref{simstudy_results} should be understood in the context that only a single thread was used to analyze each dataset. We can expect up to, for example, a four-fold speedup for each task when running the code on a two core machine with four total threads.

\section{Multivariate Spatial-Temporal Analysis of Storms}\label{datasection}

Launched in November 2016, the GOES-16 satellite now sits in geostationary orbit, allowing its advanced baseline imager (ABI) to capture images of the Earth's atmosphere in 16 wavelength bands at up to 500m resolution in space and up to 30 seconds in time. GOES 16 produces terabytes of data per day and provides a wealth of information about atmospheric processes. In this section, we perform a multivariate spatial-temporal analysis of images from two separate storms, one of which is an ordinary convective storm formed over Florida from July 22, 2018, while the other is from Hurricane Florence on September 14, 2018, the day it made landfall with the southeastern U.S. coastline. The analysis is demonstrative of the sort of comparisons that are possible with our proposed methods.

For each storm, we analyze 60 images from 4 wavelength bands separated by 1 minute. Band 1 has a native resolution of 1km, while the resolution of the other bands is 2km. We locally average the Band 1 data to coincide with the resolution of the other bands. Figure \ref{clouds_with_polygons} contains images from Band 1 at the beginning, middle, and end of the hour. A stationary model is unrealistic for the entire scene, and so we subset the data to the black polygons in Figure \ref{clouds_with_polygons}, where the stationary assumption is more tenable. Figure \ref{clouds_multiband} contains the subsetted images from the middle time point at the four wavelength bands. Table \ref{band_table} contains information about the bands selected for the analyses. We choose one band in the visible spectrum since the visible spectrum bands are all highly correlated due to the grayscale color of clouds; Band 6 contains information about particle sizes; Band 7 can be used for ice detection; and Band 9 contains information about water vapor content in the mid-troposphere \citep{abifactsheet}

The ordinary storm has observation lattice size $a = (55,57,60)$ in longitude, latitude, time, with 253{,}504 total observations; Florence has $a = (47,31,60)$, with 167{,}804 total observations. For both storms, we specify expansion factor $\tau = 1.25$, extending the lattice by 25\% in each of the three dimensions, and we again use the quasi Mat\'ern parametric spectral density. Before estimation, we subtract off the sample mean $\widehat{\mu}_k$ from each component. We use $B=20$ burn-in iterations, and we specify convergence tolerance parameter $\varepsilon = 0.005$. The ordinary storm CSD estimates converge after 6 averaging iterations, which took 3.40 hours total for burn-in and convergence. Florence converged after 15 iterations and took 1.43 hours. All computations are in the R programming language and run on an Intel Core i5-7200 CPU (2 cores, 4 threads at 2.50GHz) with 8GB memory. Vecchia's preconditioner is implemented in C++ with the Rcpp package \citep{rcpp}.

To visualize the fits, we include plots of the estimated real part of the coherences for the two storms in Figure \ref{coh}. We see that for the ordinary storm, the coherence between Bands 1,6, and 7 persists throughout the frequency domain, while Band 9 has a slight negative coherence with the other bands throughout the domain. The situation is quite different for the Hurricane Florence data; Band 9 has a small amount of positive coherence with Bands 1 and 7 and is negatively coherent with Band 6. The coherences weaken at the highest frequencies. Bands 1, 6, and 7 are weakly coherent.

\begin{table}
\centering
\begin{tabular}{ccc}
& Wavelength & Nickname \\
\hline
Band 1 & 0.47 microns & Blue Band \\
Band 6 & 2.2 microns & Cloud Particle Size Bands \\
Band 7 & 3.9 microns & Shortwave Window Band \\
Band 9 & 6.9 microns & Mid-Level Tropospheric Water Vapor Band
\end{tabular}
\caption{ \label{band_table} Information about bands used in multivariate spatial-temporal analysis. }
\end{table}

\begin{figure}
\centering
\includegraphics[width=0.8\textwidth]{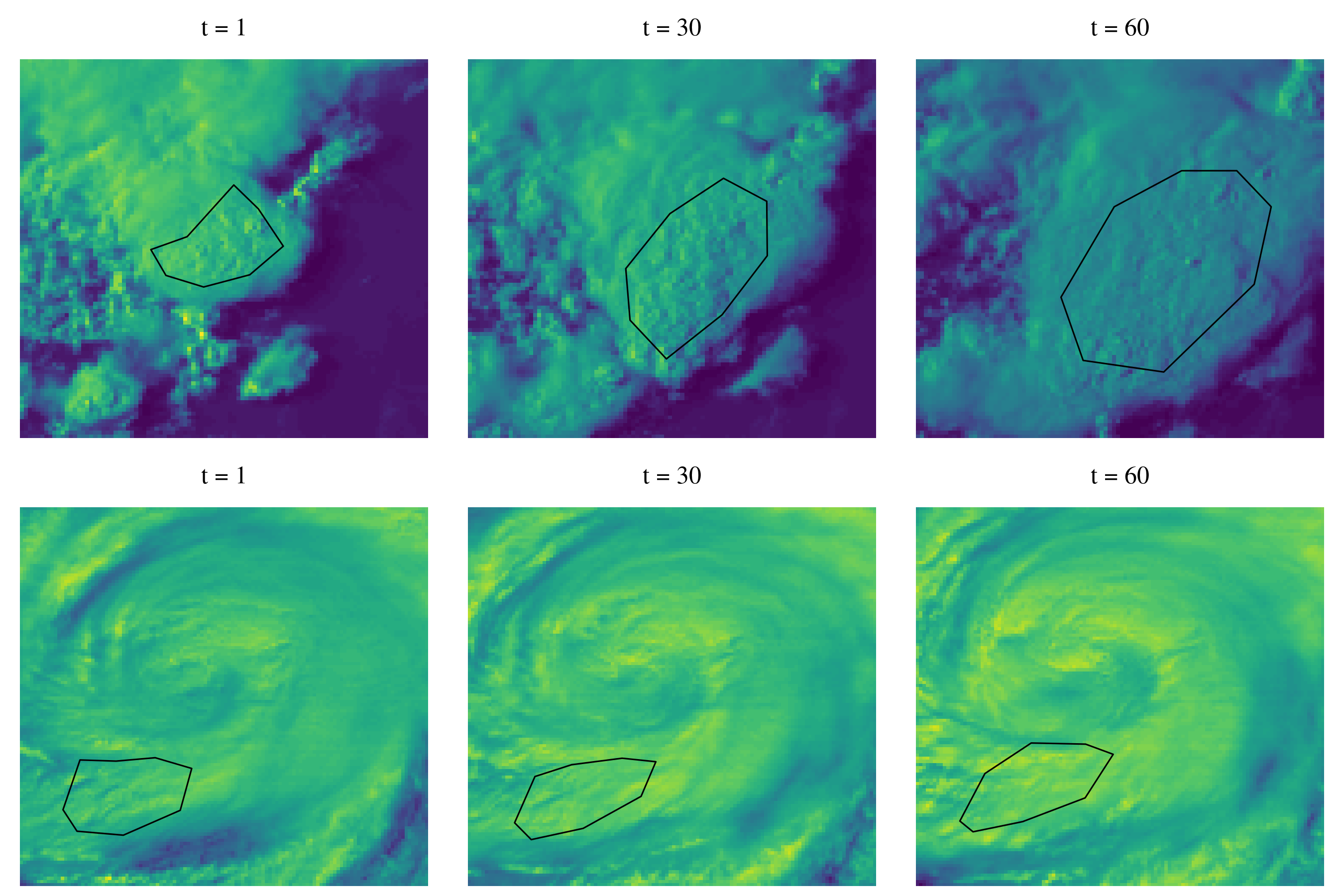}
\caption{\label{clouds_with_polygons} Data from ordinary storm (top) and from Hurricane Florence (bottom) at three time points (of 60), with subsetted data indicated by black polygons. Band 1 (``blue'' band) plotted.}
\end{figure}

\begin{figure}
\centering
\includegraphics[width=0.8\textwidth]{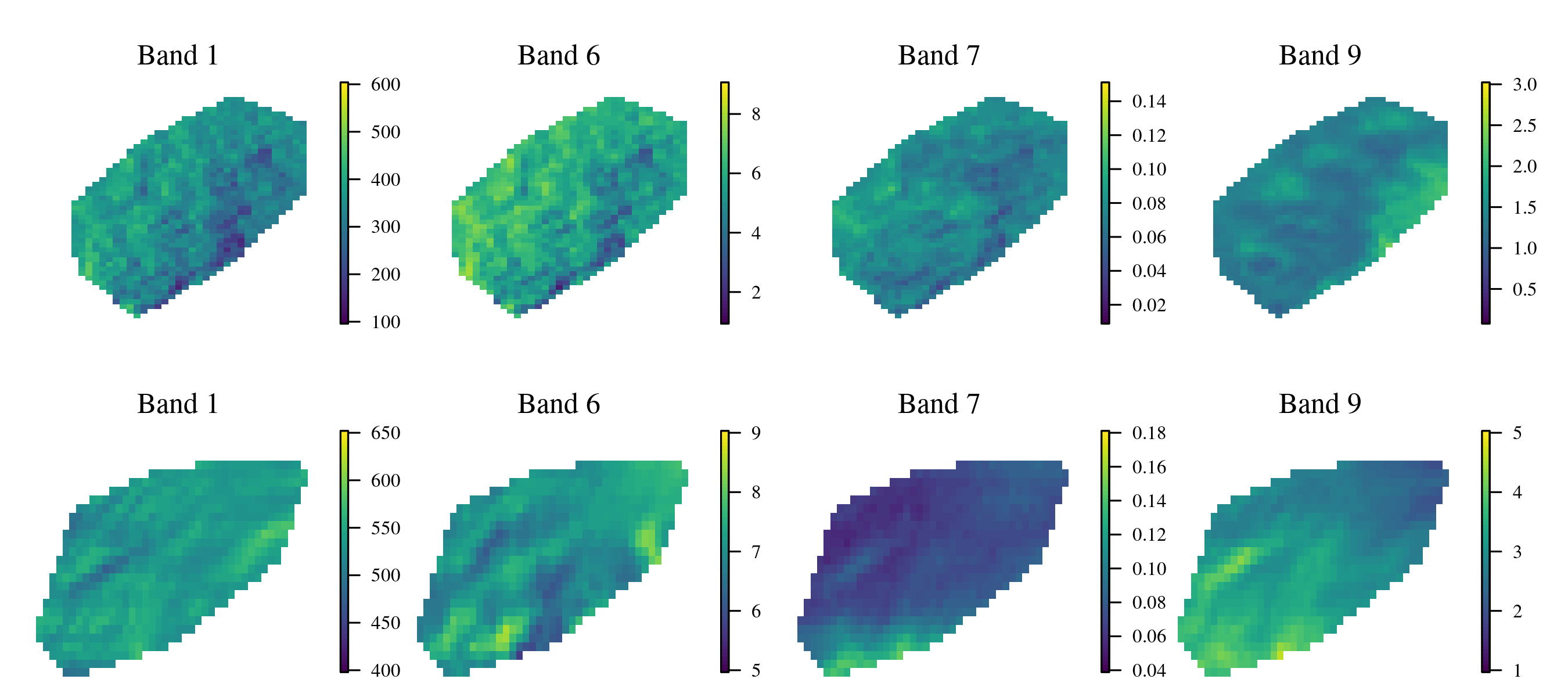}
\caption{\label{clouds_multiband} Data from ordinary storm (top) and from Hurricane Florence (bottom) at $t=30$, from Bands 1, 6, 7, and 9. }
\end{figure}

\begin{figure}
\centering
\includegraphics[width=0.29\textwidth]{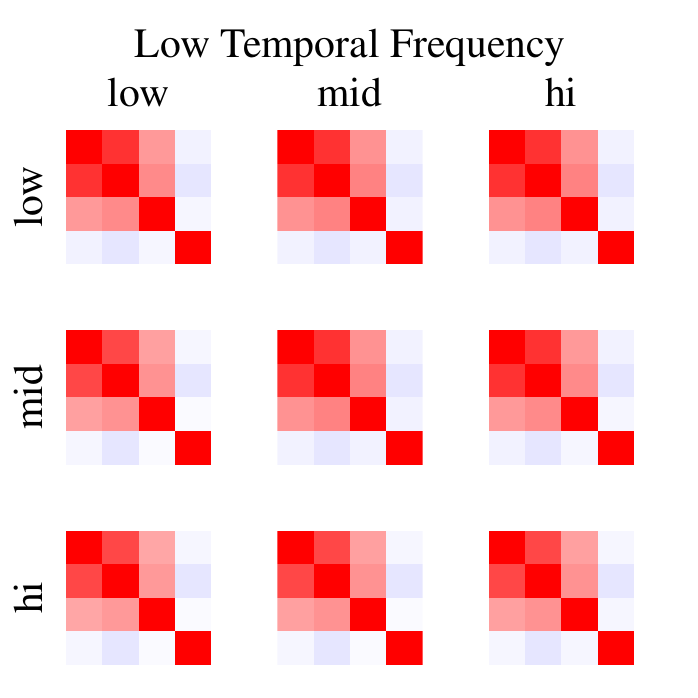}
\includegraphics[width=0.29\textwidth]{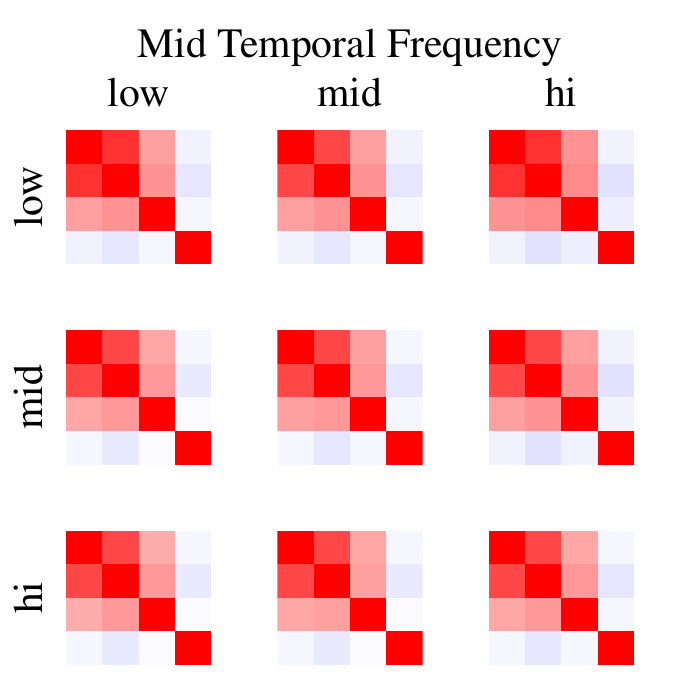}
\includegraphics[width=0.29\textwidth]{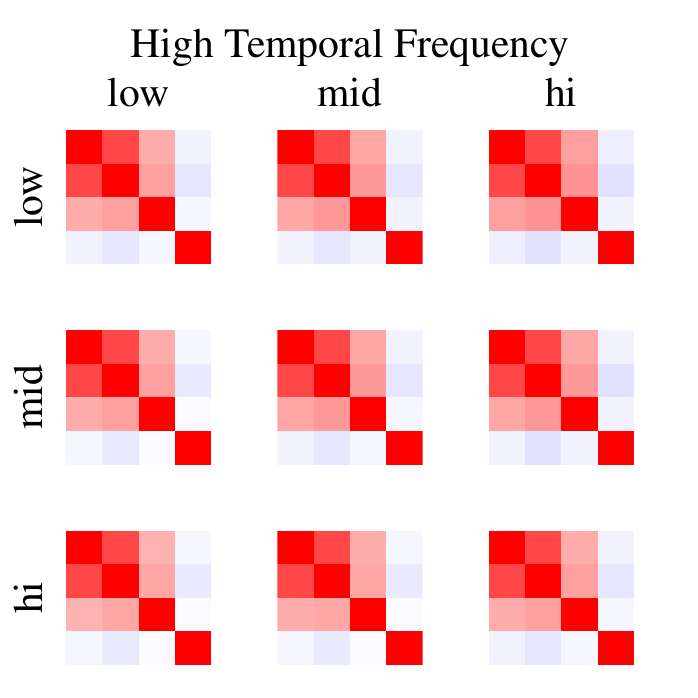}
\includegraphics[width=0.065\textwidth]{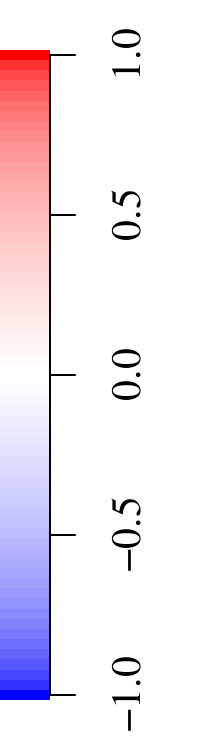}
\includegraphics[width=0.29\textwidth]{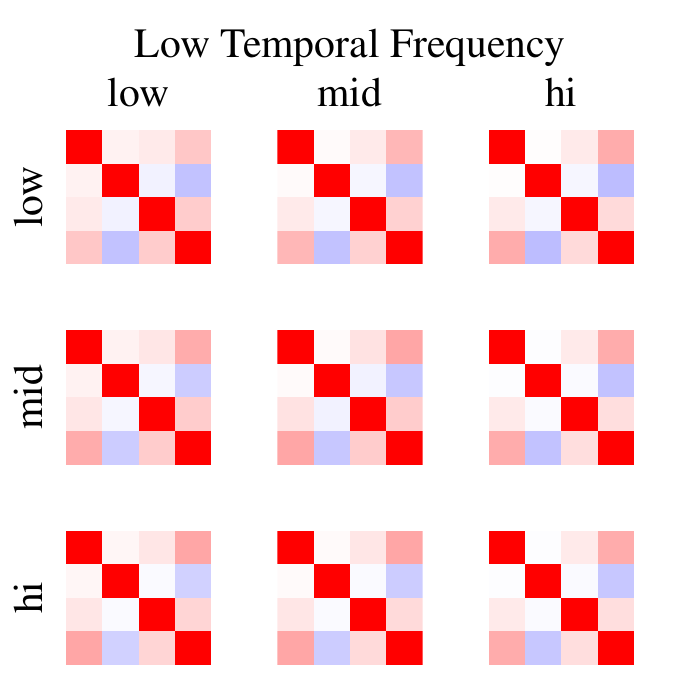}
\includegraphics[width=0.29\textwidth]{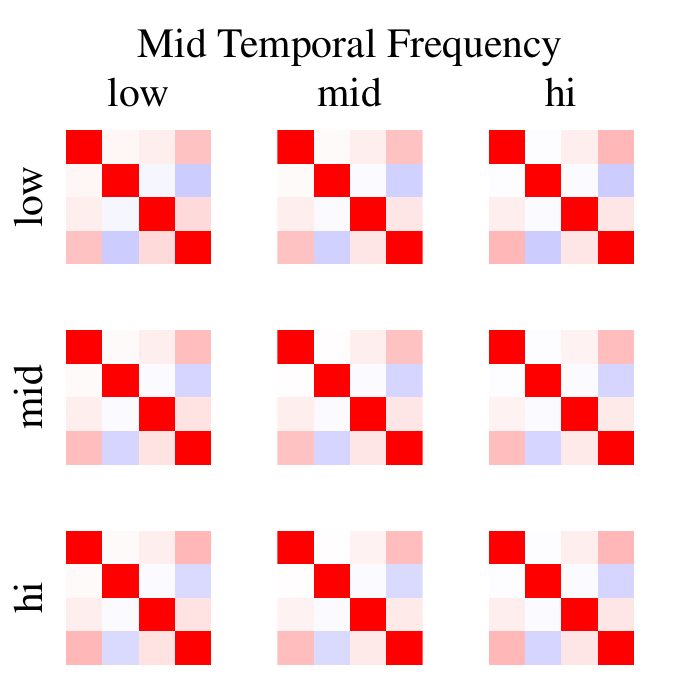}
\includegraphics[width=0.29\textwidth]{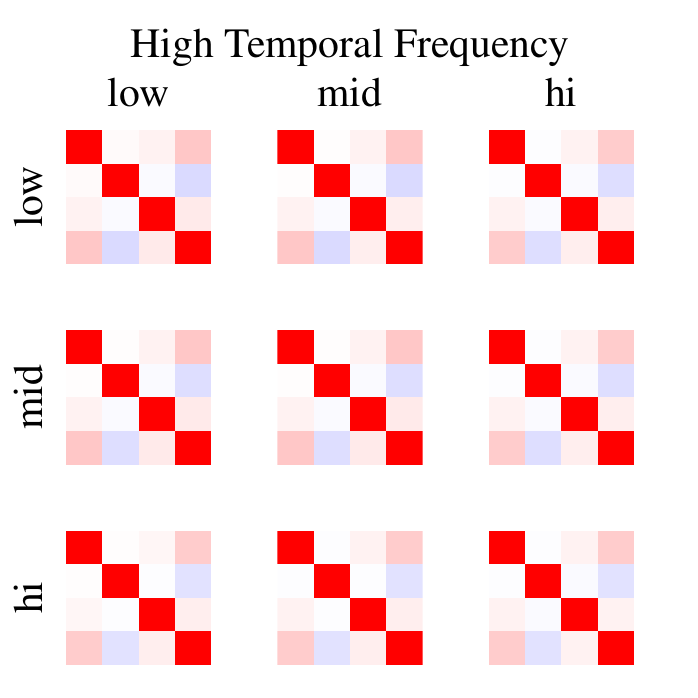}
\includegraphics[width=0.065\textwidth]{coh_legend.pdf}
\caption{\label{coh} Estimated real part of coherences for ordinary storm (top row) and Hurricane Florence (bottom row) at various frequencies $\omega = (\omega_1,\omega_2,\omega_3)$. ``Low'' refers to $\omega_j = 0$, ``mid'' to $\omega_j = 1/4$, and ``high'' to $\omega_j = 1/2$. In each $3 \times 3$ group of coherence matrices, rows are corresond to latitudinal frequencies, columns to longitudinal frequencies. }
\end{figure}

The total variation differs widely between the four bands; for example, in the ordinary storm, the estimated standard deviations are $41.15,  0.70,  0.01$ and $0.20$. These differences in variation are not necessarily a reflection of the relative importance of the bands, so instead of decomposing the estimated CSD functions $\widehat{\bm{f}}$ directly--which would undoubtedly return vectors $\bm{A}_j$ attempting to explain variation in Band 1--we decompose the normalized CSD function
\begin{align}
\widetilde{f}_{jk}(\omega) = \frac{\widehat{f}_{jk}(\omega)}{\sqrt{\widehat{C}_{jj}(0) \widehat{C}_{kk}(0) }},
\end{align}
where $\widehat{C}_{jj}(0)$ is estimated covariance function at lag zero, that is, the estimated variance. Table 1 provides a summary of the decompositions. For the ordinary storm, the $J=1$ decomposition puts most of its weight on Bands 1, 6, and 7. When $J=2$, $\bm{A}_1$ is nearly unchanged, and $\bm{A}_2$ explains variation in Band 9. For the Hurricane Florence data, when $J=1$, $\bm{A}_1$ points mostly in the direction of Bands 1, 6, and 9. When $J=2$, $\bm{A}_1$ changes little, and $\bm{A}_2$ explains common variation in Bands 1 and 6.

\begin{table}
\centering
\begin{tabular}{cr|rr|r|rr}
& \multicolumn{3}{c}{Ordinary Storm} & \multicolumn{3}{c}{Hurricane Florence} \\
 & \multicolumn{1}{c}{$J=1$} & \multicolumn{2}{c}{$J=2$} & \multicolumn{1}{c}{$J=1$} & \multicolumn{2}{c}{$J=2$} \\
\hline
Band & $\bm{A}_1$ & $\bm{A}_1$ & $\bm{A}_2$ & $\bm{A}_1$ & $\bm{A}_1$ & $\bm{A}_2$ \\
1 & $-0.498$  & $-0.492$  & $0.000$  & $-0.444$ & $-0.446$ & $-0.642$  \\
6 & $-0.634$  & $-0.634$  & $0.000$  & $0.463$ & $0.458$ & $-0.763$  \\
7 & $-0.590$  & $-0.594$  & $0.000$  & $-0.107$ & $-0.104$ & $-0.031$  \\
9 & $0.049$  & $0.054$  & $-1.000$  & $-0.760$ & $-0.762$ & $-0.064$  \\
\hline
\% explained & 43.3\% &  \multicolumn{2}{c|}{ 58.5\% } & 34.0\% &  \multicolumn{2}{c}{ 55.1\% }
\end{tabular}
\caption{\label{factor_table} Summary of factor decompositions for the two storm datasets. }
\end{table}

Lastly, in Figure \ref{storm_factors}, we plot the original ordinary storm data from $t=30$, and the expected $J=1$ and $J=2$ factor representations of each band $k$,
\begin{align}
\widehat{\mu}_k + \widehat{C}_{kk}(0)\sum_{j=1}^J A_{jk} E(W_j(x)|U).
\end{align}
We can see that the one-factor representation ignores variation in Band 9 but captures much of the variation in Bands 1, 6, and 7. This is expected since these three bands are strongly coherent across frequencies. The two-factor captures some of the variation in Band 9 but still oversmooths.

\begin{figure}
\centering
\includegraphics[width=\textwidth]{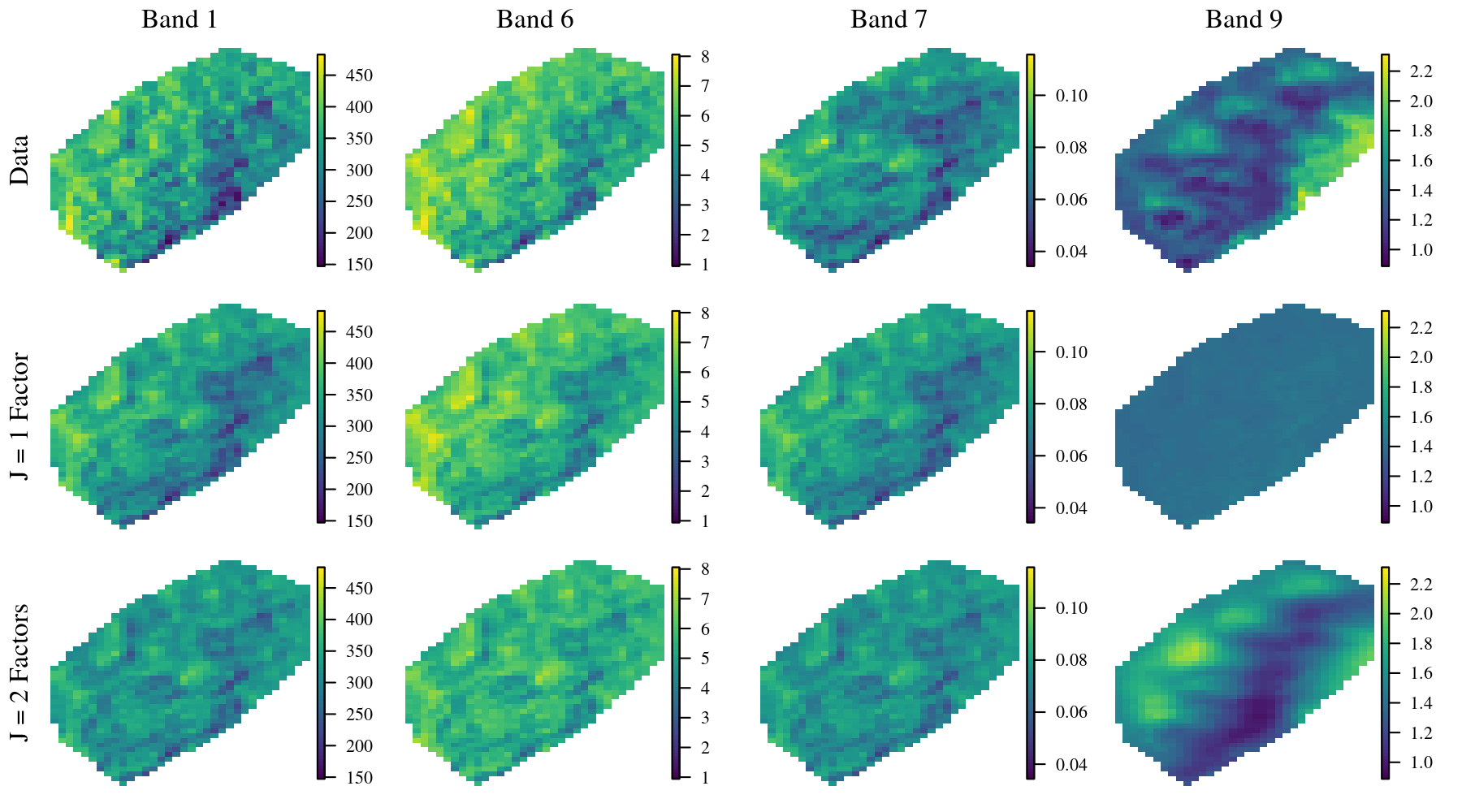}
\caption{\label{storm_factors} Data from time $t=30$ (top row), one factor representation (middle row), and two factor representation (bottom row). Factor representations use the expected values of $W_j$ given the data.}
\end{figure}

\section{Discussion}

We have introduced simple, flexible, and computationally efficient methods for estimating stationary multivariate spatial-temporal models from incomplete gridded data. The methods rely on successive imputation of data onto an expanded lattice under a model that is periodic on the expanded lattice. The simulation studies demonstrate that the periodic domain expansion is crucial for addressing edge effects; when no expansion is performed, the estimates of the spectrum are poor. The new estimates are competitive with maximum likelihood--though much faster--when there are $p=2$ multivariate components, and the estimates are faster and more accurate than maximum likelihood with $p > 2$ components. We have argued that this arises both from the computational demand of evaluating the likelihood function, and of maximizing over the large number of parameters in the multivariate Mat\'ern model.

The paper describes a method for decomposing the estimated spectrum into a linear model of coregionalization plus a residual multivariate process. The decomposition is meant as an exploratory tool for understanding the variation in the data. Finding an optimal such decomposition proved to be an interesting topic, and we have provided some theoretical results that make the numerical search over possible decompositions feasible. The factor decomposition was applied to two storm datasets, where we found that two-component decompositions explained roughly half of the variation in the data.

Only one- and two-factor decompositions were pursued here. Higher-order decompositions are in principle computationally feasible using Theorem 1 and some tedious algebra, but we leave this problem for future work. We have also not explored the predictive capabilities of the fitted models. Prediction is an interesting and important topic in multivariate spatial-temporal models, see for example \cite{zhang2015doesn}, but we have decided to focus instead on exploring the model fits and how they differ for two types of storms.

\begin{center}
\Large{\bf Acknowledgements}
\end{center}

This work was supported by the National Science Foundation under grant No.\ 1613219 and the National Institutes of Health under grant No.\ R01ES027892.

\bibliography{refs}{}
\bibliographystyle{apalike}

\appendix

\section{Supplementary Tables and Figures}

Table \ref{bandwidth_simstudy} contains simulation results for the periodic imputation methods for all four bandwidths. Figure \ref{coh_Im} contains estimated imaginary parts of the coherences for both storm datasets. Figure \ref{storm_factors_florence} shows Hurricane Florence data and estimated one- and two-factor representations of the data.

\begin{table}
\footnotesize
\centering
\begin{tabular}{cc|cccc|cccc|cccc}
& $\tau$ & \multicolumn{4}{c}{1.00} & \multicolumn{4}{c}{1.25} & \multicolumn{4}{c}{1.50} \\
$p$ & band. & 0.15 & 0.20 & 0.25 & 0.30 & 0.15 & 0.20 & 0.25 & 0.30 & 0.15 & 0.20 & 0.25 & 0.30 \\
\hline
& 0.75 &   2.726 &    2.619 &    2.567 &    2.540 &   0.373 &    0.316 &    0.303 &    0.304 &   0.330 &    0.294 &    0.293 &    0.295  \\
2 & 0.50 &   1.809 &    1.717 &    1.660 &    1.630 &   0.343 &    0.286 &    0.276 &    0.276 &   0.305 &    0.276 &    0.270 &    0.266  \\
& 0.25 &   1.362 &    1.303 &    1.248 &    1.206 &   0.305 &    0.262 &    0.253 &    0.253 &   0.270 &    0.239 &    0.241 &    0.243  \\
& minutes & 0.085 &  0.075 &  0.080 &  0.078 & 0.289 &  0.263 &  0.237 &  0.227 & 0.354 &  0.316 &  0.296 &  0.261 \\
\hline
& 0.75 &   2.776 &    2.630 &    2.544 &    2.494 &   0.479 &    0.410 &    0.391 &    0.379 &   0.449 &    0.383 &    0.363 &    0.358  \\
3 & 0.50 &   1.992 &    1.856 &    1.798 &    1.767 &   0.449 &    0.376 &    0.349 &    0.337 &   0.410 &    0.346 &    0.329 &    0.328  \\
& 0.25 &   1.547 &    1.428 &    1.353 &    1.300 &   0.423 &    0.338 &    0.315 &    0.306 &   0.388 &    0.323 &    0.312 &    0.302  \\
& minutes & 0.091 &  0.082 &  0.082 &  0.081 & 0.384 &  0.326 &  0.282 &  0.273 & 0.413 &  0.354 &  0.327 &  0.296  \\
\hline
& 0.75 &   3.105 &    2.898 &    2.754 &    2.654 &   0.569 &    0.479 &    0.448 &    0.429 &   0.509 &    0.439 &    0.414 &    0.407  \\
4 & 0.50 &   2.238 &    2.092 &    1.964 &    1.888 &   0.532 &    0.441 &    0.399 &    0.398 &   0.475 &    0.405 &    0.390 &    0.383  \\
& 0.25 &   1.806 &    1.604 &    1.480 &    1.412 &   0.501 &    0.405 &    0.375 &    0.372 &   0.447 &    0.381 &    0.361 &    0.358  \\
& minutes & 0.102 &  0.089 &  0.092 &  0.090 & 0.487 &  0.424 &  0.380 &  0.363 & 0.593 &  0.489 &  0.445 &  0.416  \\
\hline
\end{tabular}
\caption{\label{bandwidth_simstudy} Simulation results for all four smoothing kernel bandwidths for periodic imputation methods.}
\end{table}

\begin{figure}
\centering
\includegraphics[width=0.29\textwidth]{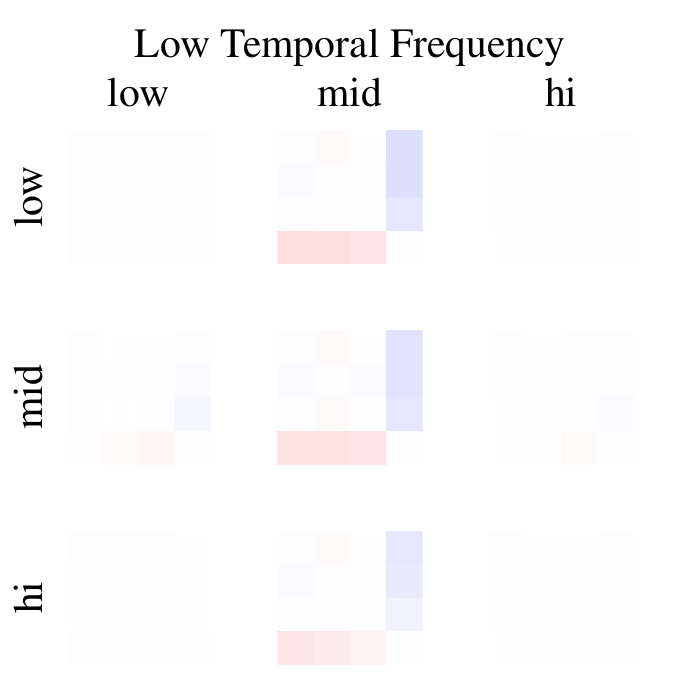}
\includegraphics[width=0.29\textwidth]{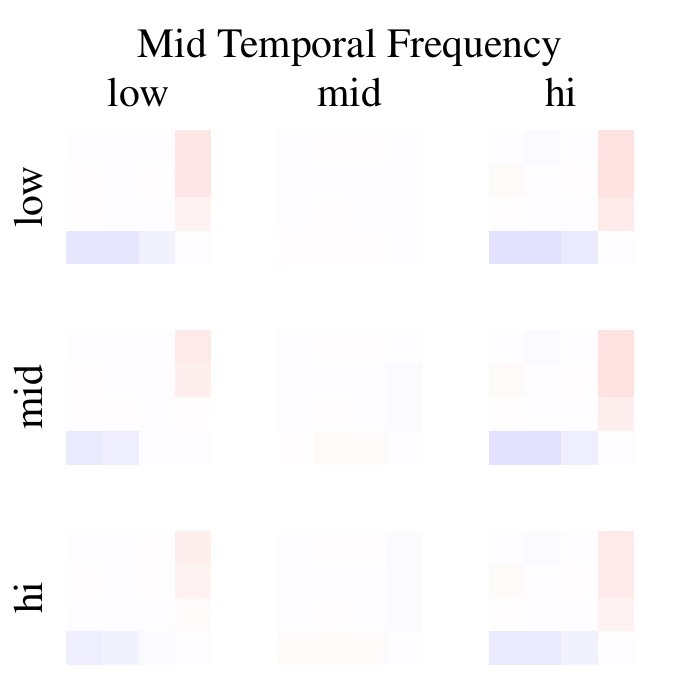}
\includegraphics[width=0.29\textwidth]{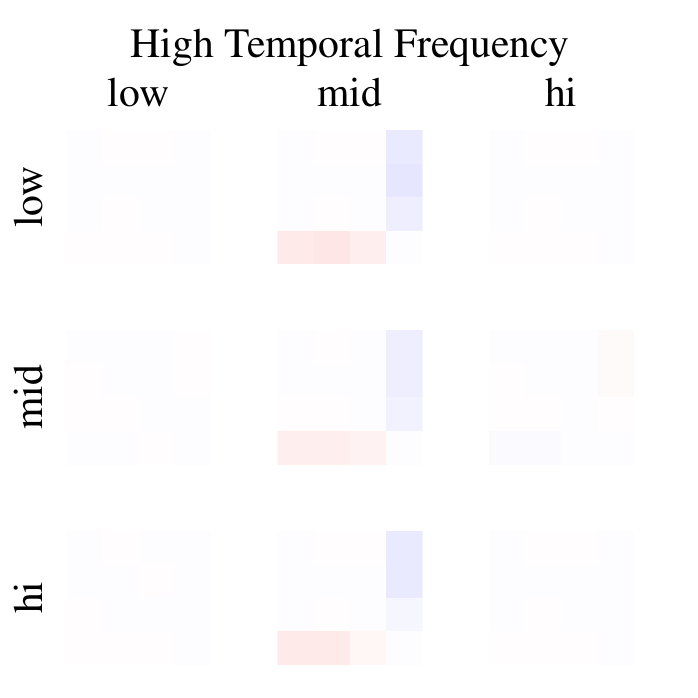}
\includegraphics[width=0.065\textwidth]{coh_legend.pdf}
\includegraphics[width=0.29\textwidth]{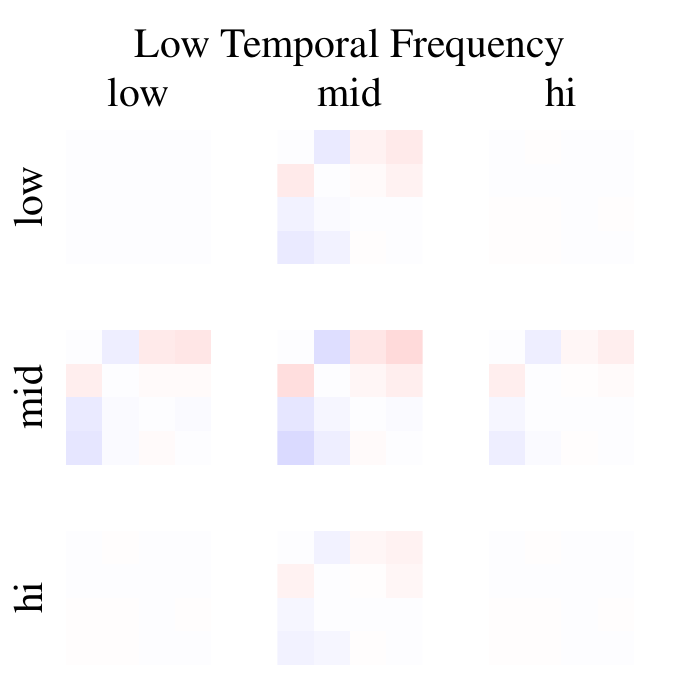}
\includegraphics[width=0.29\textwidth]{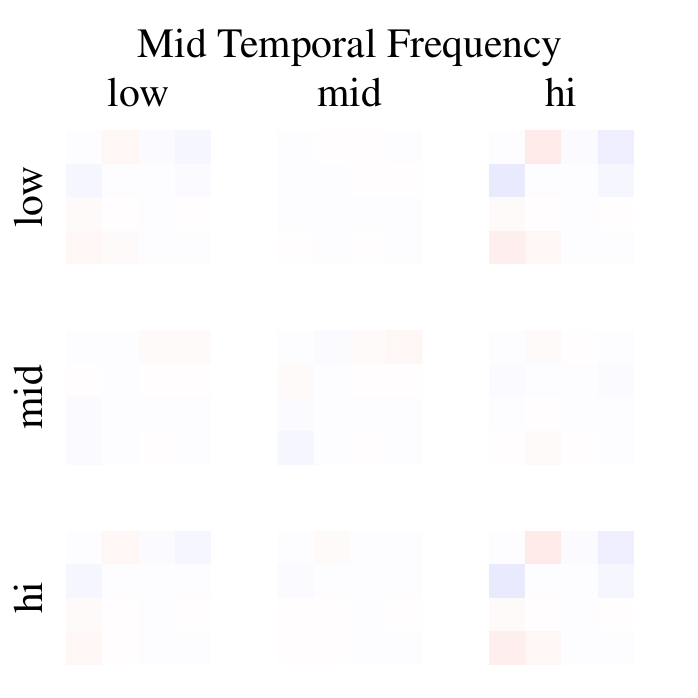}
\includegraphics[width=0.29\textwidth]{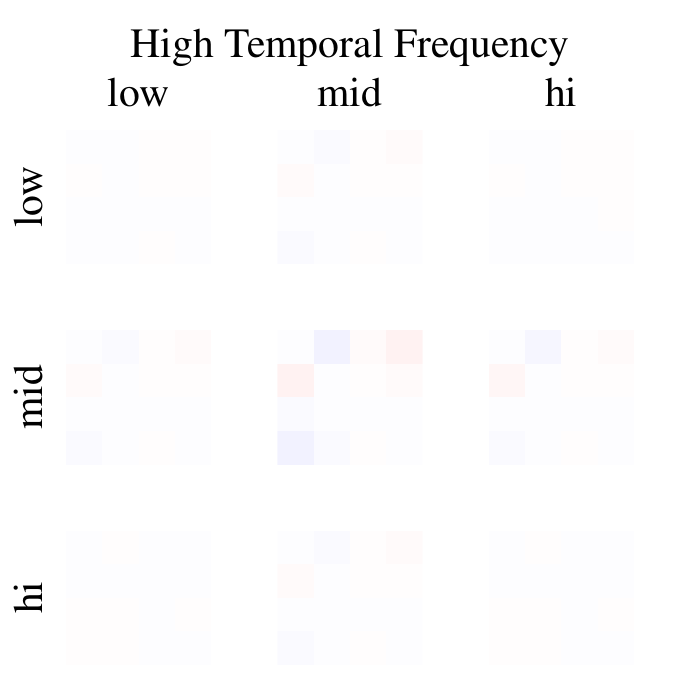}
\includegraphics[width=0.065\textwidth]{coh_legend.pdf}
\caption{\label{coh_Im} Estimated imaginary part of coherences for ordinary storm (top row) and Hurricane Florence (bottom row) at various frequencies $\omega = (\omega_1,\omega_2,\omega_3)$. ``Low'' refers to $\omega_j = 0$, ``mid'' to $\omega_j = 1/4$, and ``high'' to $\omega_j = 1/2$. In each $3 \times 3$ group of coherence matrices, rows are corresond to latitudinal frequencies, columns to longitudinal frequencies. }
\end{figure}

\begin{figure}
\centering
\includegraphics[width=\textwidth]{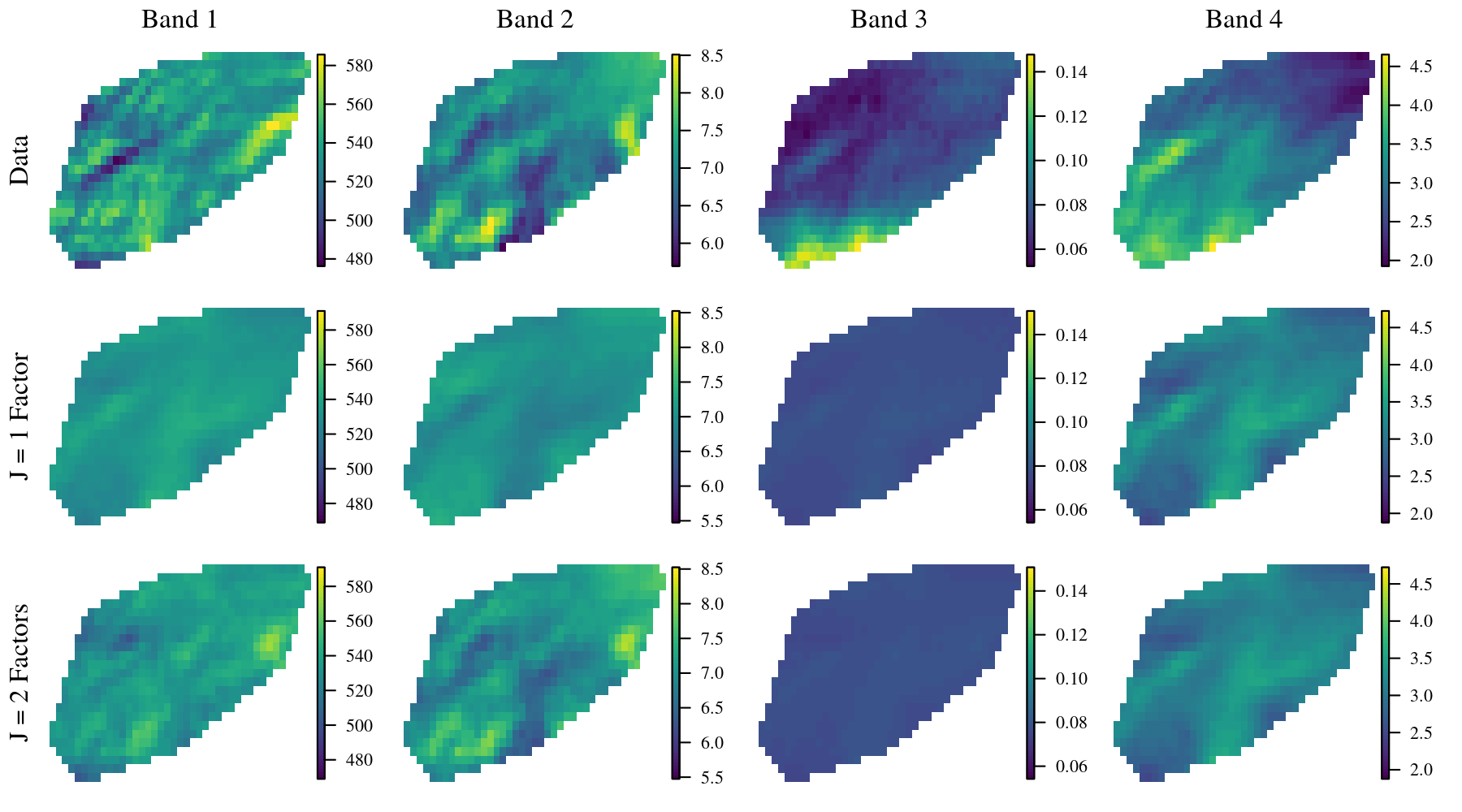}
\caption{\label{storm_factors_florence} Hurricane Florence data from time $t=30$ (top row), one factor representation (middle row), and two factor representation (bottom row). Factor representations use the expected values of $W_j$ given the data.}
\end{figure}

\end{document}